\documentclass[conference]{IEEEtran}

\IEEEoverridecommandlockouts

\usepackage{cite}
\usepackage{nicefrac}
\usepackage{amsmath,amssymb,amsfonts}
\usepackage{amsthm}
\usepackage[algo2e]{algorithm2e} 
\usepackage{graphicx}
\usepackage{textcomp}
\usepackage{xcolor}
\usepackage{physics}
\usepackage{amsmath}
\usepackage[utf8]{inputenc}
\usepackage[numbers]{natbib}
\usepackage{amsmath}
\usepackage{graphicx}
\usepackage{subcaption}
\usepackage{tikz}
\usetikzlibrary{quantikz}
\usetikzlibrary{arrows}
\usepackage{adjustbox}
\usepackage{graphicx}
\usepackage{bm}
\usepackage{xfrac}
\usepackage{multirow}
\usepackage{booktabs}
\usepackage{url}
\usepackage{balance}
\usepackage{algorithm}
\usepackage{algpseudocode}

\newcommand{\bI}{\boldsymbol{I}}

\newcommand{\bR}{\boldsymbol{R}}

\newcommand{\bx}{\boldsymbol{x}}

\newcommand{\cD}{\mathcal{D}}

\newcommand{\cO}{\mathcal{O}}

\newcommand{\subjectto}

\usepackage{amsmath}
\usepackage{algorithm2e}
\def\BibTeX{{\rm B\kern-.05em{\sc i\kern-.025em b}\kern-.08em
    T\kern-.1667em\lower.7ex\hbox{E}\kern-.125emX}}

\newtheorem{theorem}{Theorem}
\newtheorem{proposition}[theorem]{\textbf{Lemma}}%

\usepackage{algorithmicx}

\begin{document}

\title{A Deployable Quantum Access Points Selection  Algorithm for Large-Scale Localization}
\author{\IEEEauthorblockN{Ahmed Shokry}
\IEEEauthorblockA{\textit{Computer Science and Engineering} \\
\textit{Pennsylvania State University}\\
PA, USA \\
ahmed.shokry@psu.edu}
\and
\IEEEauthorblockN{Moustafa youssef}
\IEEEauthorblockA{\textit{Computer Science and Engineering} \\
\textit{American University in Cairo}\\
Cairo, Egypt \\
moustafa-youssef@aucegypt.edu}
}

\maketitle

\begin{abstract}
Effective access points (APs) selection is a crucial step in  localization systems. It directly affects both localization accuracy and computational efficiency. Classical APs selection algorithms are usually computationally expensive, hindering the deployment of localization systems in a large worldwide scale. 

In this paper, we introduce a quantum APs selection algorithm for large-scale localization systems. The proposed algorithm leverages quantum annealing to eliminate redundant and noisy APs. We explain how to formulate the APs selection problem as a quadratic unconstrained binary optimization (QUBO) problem, suitable for quantum annealing, and how to select the minimum number of APs that maintain the same overall localization system accuracy as the complete APs set. Based on this, we further propose a logarithmic-complexity algorithm to select the optimal number of APs. 

We implement our quantum algorithm on a real D-Wave Systems quantum machine and assess its performance in a real test environment for a floor localization problem. Our findings reveal that by selecting fewer than 14\% of the available APs in the environment, our quantum algorithm achieves the same floor localization accuracy as utilizing the entire set of APs and a superior accuracy over utilizing the reduced dataset by  classical APs selection counterparts. Moreover, the proposed quantum algorithm achieves more than an order of magnitude speedup over the corresponding classical APs selection algorithms, emphasizing the  efficiency of the proposed quantum algorithm for large-scale localization systems.

\end{abstract}


\section{Introduction}
Large-scale fingerprinting localization systems are commonly employed both indoors \cite{youssef2015towards,shokry2017tale} and outdoors \cite{shokry2018deeploc, ibrahim2011cellsense} due to their high accuracy. These systems utilize the Radio Frequency (RF) technology and follow a two-stage process. In the offline phase, they create an ``RF fingerprint'' by capturing the received signal strength (RSS) signature from different access points (APs)  at different  locations in the area of interest and store this information in a database. These APs can include WiFi access points, cell towers, or Bluetooth beacons. In the online tracking phase, the RSS received from the APs at the user's device with unknown location is compared with the fingerprint samples in the database to identify the location closest, in the RSS space, to the unknown location.

As the fingerprinting localization techniques need to compare online RSS measurements with offline ones at every fingerprint location, their time and space complexity increase with the increase in number of APs. Classical algorithms try to infer a promising reduced subset of APs by adding or removing APs and comparing the results~\cite{jovic2015review}. But ultimately, when the number of APs is large (especially in a worldwide localization system), this exhaustive exploration of all combinations of APs becomes computationally prohibitive and a trade-off remains between the accuracy of selecting the effective APs and the effort spent finding the best subset of APs. This complexity does not allow for scaling  of classical APs selection techniques to support large-scale indoor/outdoor localization systems on a global scale, e.g. in IoT environments where there may be a significant number of APs with wireless interfaces.

Recently, there have been successful proposals of quantum fingerprint matching algorithms in the literature \cite{quantum_arx, quantum_vision, device_indp_q, SHOKRY2023, quantum_lcn, quantum_qce, shokry2023qradar}. These techniques leverage universal \textbf{gate-based} quantum machines to achieve exponential saving in matching time and space in terms of the number of APs and in terms of the number of fingerprint locations~\cite{zook2023quantum}. However, due to quantum decoherence~\cite{schlosshauer2005decoherence}, the current gate-based quantum machines can only handle a limited number of qubits accurately, which restricts the deployment of large-scale localization systems \textit{in the near future} on them.

In this paper, we propose a novel approach for APs selection that can help in reducing the computational power of fingerprinting-based localization algorithm. Specifically, we propose a quantum APs selection algorithm that leverages quantum annealing to eliminate redundant and noisy APs which results in improving the overall localization accuracy. We present a formulation of the APs selection problem as a Quadratic Unconstrained Binary Optimization (QUBO) problem~\cite{QUBO1, QUBO2}, suitable for processing on quantum annealers, with the aim of selecting the powerful APs that improves overall localization system efficiency. 
The proposed quantum algorithm utilizes binary variables to denote APs selection and formulates an objective function that tries to select APs that have a strong effect on the location information (i.e. maximize important APs) and have a weak correlation between each others (i.e. minimize redundancy).
Furthermore, we present an algorithm designed to identify the minimum set of powerful APs necessary to enhance localization accuracy. 
Our proposed algorithms leverages the real D-Wave Systems quantum machines, which are  \textbf{\textit{specialized practical}} quantum computers for solving quantum annealing problems~\cite{Dwave}. 

Results from deploying the proposed quantum algorithm and evaluating its performance in a real testbed for a floor identification task reveal that by selecting fewer than 14\% of the available APs in the environment, our quantum algorithm achieves the same floor localization accuracy as utilizing the entire set of APs and a superior accuracy over utilizing the reduced dataset by classical APs selection counterparts. Moreover, the proposed quantum algorithm achieves a speedup of more than an order of magnitude over the corresponding classical APs selection algorithms.

The rest of the paper is organized as follows: Section~\ref{sec:background} gives a background on  quantum computing and the floor localization problem. Section~\ref{sec:method} provides the details of our quantum APs selection algorithm. We present the details of the implementation and evaluation results in Section~\ref{sec:evaluation}. Finally, Section \ref{sec:conclude} concludes the paper.

\section{Background}
\label{sec:background}
In this section, we start by a background on quantum computing and quantum annealing. Then, we give a brief introduction to the floor estimation problem, which is the case study we apply our APs selection algorithm to.

\subsection{Quantum Computing and Quantum Annealing}

Quantum computing represents a paradigm shift in computational theory and practice, harnessing the principles of quantum mechanics to enable the solution of complex problems that classical computers struggle with. At its core, quantum computing leverages quantum bits (qubits) that, unlike classical bits, can exist in a \textit{superposition} state of 0 and 1 at the same time, exponentially increasing computational capacity~\cite{nielsen2002quantum}. 

Quantum \textit{entanglement} is another phenomenon where two or more qubits become correlated and the state of one qubit instantaneously influences the state of another, regardless of the distance between them~\cite{nielsen2002quantum}. This entanglement property enables the creation of highly-correlated quantum states, a phenomenon that is fundamental to the efficiency of quantum algorithms. The power of quantum computing lies in its ability to perform parallel computations due to superposition and to exploit quantum entanglement for correlations between qubits.

Quantum annealing is a specialized approach within quantum computing~\cite{qa}, tailored for solving optimization problems which involve finding the best solution from a vast number of possibilities. Quantum annealing has been used in various fields including finance~\cite{finance} and artificial intelligence~\cite{nath2021review}. It follows the principles of adiabatic quantum computation~\cite{qa}, where a system is slowly evolved from a simple Hamiltonian (i.e. energy function) to a complex one, reaching the ground state that encodes the solution. Adiabatic quantum computation involves initializing a quantum system in the ground state of a simple Hamiltonian ($H_0$), and slowly transforming it into the ground state of a more complex final Hamiltonian ($H_f$). This adiabatic evolution allows the system to settle into the state representing the optimal solution to an optimization problem. The simple Hamiltonian $H_0$ represents the initial state of the quantum system. It is carefully chosen such that its ground state is easy to prepare, typically a state where each qubit is in a simple state. The final Hamiltonian $H_f$ encodes the optimization problem. Its ground state contains information about the optimal solution to the problem. The adiabatic quantum computation ensures that the system evolves into this final ground state, capturing the optimal solution.

The evolution of a quantum system during annealing is described by the Schrödinger equation~\cite{griffiths2018introduction}. For a quantum state $| \psi(t) \rangle$, the time-dependent Schrödinger equation is given by:
\begin{equation}
    i \hbar \frac{\partial}{\partial t} | \psi(t) \rangle = H(t) | \psi(t) \rangle
\end{equation}
Here, $i$ is the imaginary unit, $\hbar$ is the reduced Planck constant, $H(t)$ is the time-dependent Hamiltonian operator representing the system's energy, and $\ket{\psi(t)}$ is the state of the quantum system at time $t$. The Hamiltonian operator is typically expressed as a combination of the simple Hamiltonian ($H_0$) and the final Hamiltonian ($H_f$). The annealing schedule is controlled by a parameter $s(t)$, which varies from 0 to 1 as the system evolves. The time-dependent Hamiltonian is then defined as:
\begin{equation}
    H(t) = (1 - s(t)) H_0 + s(t) H_f
\end{equation}

The objective is to start with the simple Hamiltonian ($H_0$) and gradually transition to the final Hamiltonian ($H_f$) to guide the system towards the optimal solution. This process exploits quantum phenomena such as superposition and entanglement to explore a vast solution space more efficiently than classical methods.

The previous construction plays an important role not only in the realm of quantum mechanics, but also in the field of optimization. Nowadays, many real-world optimization problems can be expressed as QUBO problems~\cite{QUBO1, QUBO2}, a class of problems that involves finding the minimum (or maximum) of a quadratic objective function subject to binary constraints. Formally, a QUBO problem can be defined as:

\begin{equation}
    \text{Minimize } f(\mathbf{x}) = \sum_{i,j} Q_{ij}x_i x_j + \sum_i P_i x_i
\end{equation}

where $\mathbf{x}$ is a binary vector, $Q_{ij}$ are coefficients representing the quadratic terms, and $P_i$ are coefficients representing the linear terms.
QUBO problems can be reformulated as an Ising Hamiltonian, which is a mathematical model that provides a simplified representation of magnetic interactions in a system of atomic spins~\cite{ISING}. The basic idea is to map the QUBO problem onto a physical system described by the Ising model, where the spins of magnetic particles represent variables and their interactions encode the problem's constraints and objectives. In this mapping, binary variables correspond to spins, and the QUBO objective function maps onto the energy function of the Ising model (i.e. Hamiltonian). This transformation allows quantum annealers, which naturally operate with Ising models, to be employed in solving QUBO problems. Quantum annealers offer a \textit{practical} avenue for solving QUBO problems. The encoding of the QUBO problem into the Ising model allows quantum annealers to explore the solution space through quantum fluctuations, seeking the optimal configuration.

\subsection{The Floor Estimation Problem}
The floor estimation problem~\cite{main_classical_paper, fi_fingerprint_2018, fi_fingerprint_ka, VIFI_fingerprint_2019} relies on utilizing pre-installed signal sources (e.g., WiFi APs or cellular base stations) to determine users' floor. Typically, the floor estimation algorithms operate in two phases: (a) Offline calibration phase: During this phase, signal information such as RSS and signal source IDs are collected at various known locations on each floor to construct a fingerprint. Each fingerprint is stored as vectors, where each entry represents the RSS from one AP, along with the corresponding floor number. (b) Online tracking phase: In this phase, the information from online signal sources is compared to the fingerprint vectors, and the closest floor in the signal space is reported as the estimated floor.

\section{The Proposed Quantum APs Selection Algorithm}
\label{sec:method}

\begin{table}[!t]\vspace{+0.3cm} \caption{ Table of notations.}\label{tab:notations}
\centering
 \begin{tabular}{|l | l|} 
 \hline
 \textbf{Symbol} & \textbf{Description} \\ [0.5ex] 
 \hline\hline
 $m$ & Number of fingerprint samples \\   \hline

$n$ & Number of APs \\  \hline

$r$ & A fingerprint RSS sample vector \\  \hline

$b$ & Number of bins in the RSS range \\  \hline 

$l$ & A fingerprint floor label \\  \hline

$f$ & Number of distinct locations in the fingerprint \\  \hline 

$\cD=\lbrace(r,l)\rbrace$ & Original fingerprint \\  \hline

$\cD_S=\lbrace (r_S,l)\rbrace$ & Reduced fingerprint after selecting effective S APs \\   \hline

$I_i$ & Importance of AP $i$ \\  \hline 

$R_{ij}$ & Mutual correlation between two APs (redundancy)\\   \hline

$\alpha$ & Balancing parameter (importance and redundancy) \\  \hline

$\bx^*$ & A binary vector indicating the selected APs \\  \hline

$k = \lVert\bx^*\rVert_1$  & Number of selected APs \\ [1ex] 
 \hline
 \end{tabular}
\end{table}

In this section, we describe our proposed APs selection algorithm.
 We start with formulating APs selection problem as a QUBO problem. Then, we describe how to select the minimum number of APs that achieve the best localization accuracy. Based on this, we present our APs selection algorithm. 
 Table~\ref{tab:notations}  summarizes the notations used in this section. 
 
\subsection{QUBO Formulation}
\label{sec:qfs}

\begin{figure*}[!t]
  \begin{subfigure}{0.49\textwidth}
\includegraphics[width=\linewidth]{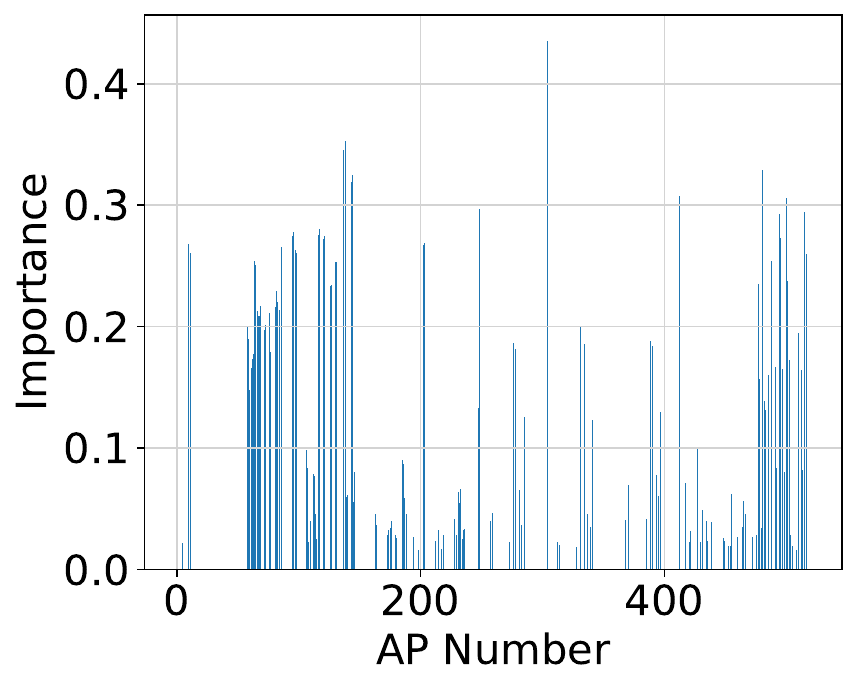}
    \caption{Importance} \label{subfig:imp}
  \end{subfigure}%
  \begin{subfigure}{0.49\textwidth}
    \includegraphics[width=\linewidth]{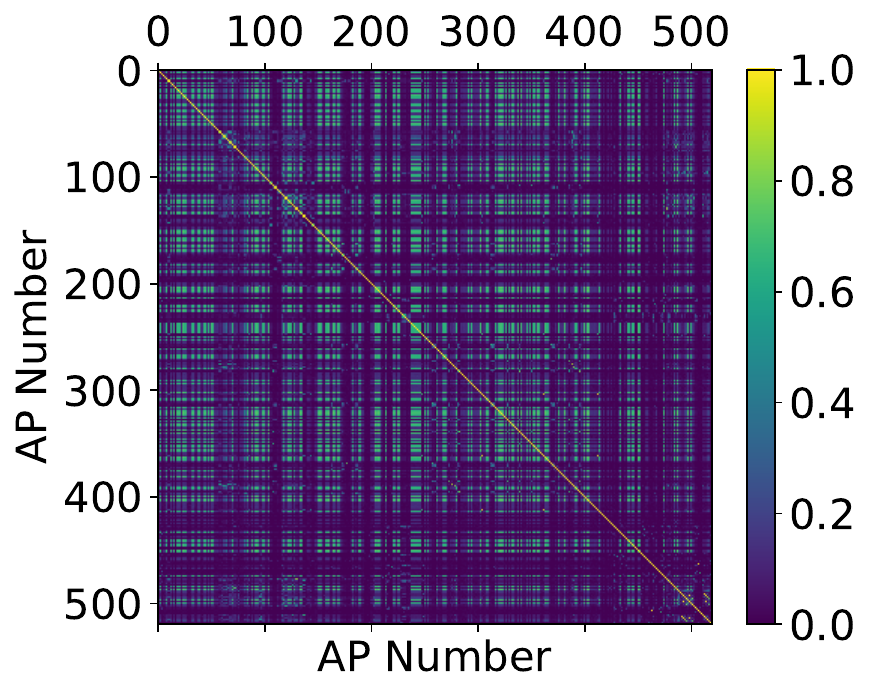}
    \caption{Redundancy} \label{subfig:red}
  \end{subfigure}%
\caption{Importance and Redundancy in the full dataset. The importance reflects the influence of each AP on estimating the location. The Redundancy reflects APs inter-dependency. } \label{fig:red_imp_full1}
\end{figure*}

Assume a localization problem on a fingerprint with $m$ fingerprint samples and $n$ APs,  $\cD=\lbrace(r^i,l^i)\rbrace_{i\in \{1, 2, \dots, m\}}$ with $n$-dimensional APs received signal strength (RSS) vector $r^i \in \mathbb{R}^n$ and floor information labels $l^i \in \mathbb{N}$ for all $i\in\{1,2,3,\dots, m\}$.
The APs selection problem tries to find a subset of APs $S\subset \{1,2,3,\dots,n\}$, such that the reduced dataset $\cD_S=\lbrace (r^i_S,l^i)\rbrace_{i\in \{1,2,3,\dots,m\}}$ with the reduced dimension vector $r_S$ referring to the RSS from the selected APs only; leads to comparable localization accuracy as the full original RSS vectors $r$ with all APs.
Typically, this subset is found by searching for the ``optimal'' subset of APs that lead to the best localization accuracy. 

Our APs selection problem goal is to select APs that have the highest influence on estimating the location (i.e., highest importance). It also tries to reduce the APs that depend on each others (i.e., lowest redundancy) as shown in Figure~\ref{fig:red_imp_full1}. We propose a quantum algorithm that targets maximizing the importance and minimize the redundancy of APs. Specifically, we represent the APs subset $S$ as a binary indicator vector $\bx=(x_1,\dots,x_n)\in\lbrace 0,1\rbrace^n$, such that AP $i\in S$ if and only if $x_i=1$ for all $i \in \{1,2,3,\dots,n\}$. Therefore, the optimal subset of APs, $x^*$ can be obtained as

\begin{equation}\label{eq:qfs}
	\bx^* = \underset{\bx\in\lbrace 0,1\rbrace^n}{\arg\min} \, Q(\bx,\alpha)
\end{equation}

Where our QUBO objective function, $Q(\bx,\alpha)$, is defined as

\begin{equation}\label{eq:Q}
	Q(\bx,\alpha) = - \alpha \sum_{i=1}^n I_i x_i + (1-\alpha) \sum_{i,j=1}^n R_{ij} x_i x_j,
\end{equation}

where the first sum of linear terms represents the individual contributions from APs (i.e., importance, $I_i$) and the second sum contains quadratic interaction terms between APs pairs (i.e., redundancy, $R_{ij}$). The user-defined parameter $\alpha \in [0,1]$ balances the influence of the two terms: the importance and the redundancy.
Finding the set of binary variables $\bx_i$ that minimizes this objective function is equivalent to finding the  optimal set of APs.

In the balance of this section, we will give the details of how Importance and Redundancy are calculated in our algorithm. Finally, we will provide an algorithm for controlling the number of selected APs.

\subsubsection{Importance}
We define the importance of each AP $i$ as the Cramér's Variation (Cramér's V) measure of the signal strength association with the floor labels 
\cite{cramer1999mathematical}. This measure is an extension of chi-squared test, which is used to determine if there is a significant association between two categorical variables. The Cramer's V measure further standardizes the chi-squared statistic to a value between 0 and 1, making it easier to interpret; regardless of the specific variables distributions or scales.

Hence, our Importance measure ranges from 0 to 1, where 0 indicates no association between the AP RSS and floor numbers, and 1 indicates a strong association. 

Specifically, the Importance between an AP $i$ RSS values and the fingerprint location labels $l$ can then be calculated as
\begin{equation}
I_i = \sqrt{\frac{\chi^2}{m \times \min(\text{{b-1}}, \text{{f-1}})}}    
\end{equation}
Where $\chi^2$ is the chi-square statistic which measures the discrepancy 
between the observed frequencies of AP $i$ RSS with the output labels, and the frequencies that would be expected if the two variables were independent. $m$ is the number of fingerprint samples, $f$ is the number of distinct labels, and $b$ is the number of possible RSS values. 

Figure~\ref{subfig:imp} shows the importance of all APs within the environment. The figure shows that the fingerprint contains many APs with low importance. These APs can be removed to reduce computational complexity without affecting the accuracy of the localization system.

\subsubsection{Redundancy}
We define the redundancy between each pair of APs, $i$ and $j$, as the strength of the absolute correlation between them.

\begin{equation}\label{eq:Rij}
	R_{ij} = \text{$|$Corr(}r_i,r_j\text{)$|$} \geq 0
\end{equation}

where $\bR\in\mathbb{R}^{n \times n}$ is symmetric and positive semi-definite matrix.
This matrix represents the absolute mutual correlation $|Corr(r_i,r_j)|$ among the  APs RSS and therefore measures their redundancy.
For $i=j$, we set $R_{ii}=0$. 
Without loss of generality,  we use Pearson correlation as the measure of redundancy (\textit{Corr(.)})~\cite{pearson}. 
Figure~\ref{subfig:red} shows the correlation between each pair of APs within the environment. The figure shows that the fingerprint contains many APs pairs that depend on each other (color towards yellow). This also adds to the computational cost on the localization system without a real gain on accuracy.

\subsection{Objective Function in Matrix Form}
The QUBO objective function in Eq~\ref{eq:Q} can be written a matrix form $\mathbf{P}(\alpha)\in\mathbb{R}^{n \times n}$ for better mathematical manipulation, where a diagonal item $\mathbf{P}_{ii}$ presents the importance of each AP $i$ (i.e. $I_i$), and a non-diagonal item $\mathbf{P}_{ij}$ is the redundancy (i.e. $R_{ij}$). Hence, the optimal subset of APs in Eq~\ref{eq:qfs} can be expressed in matrix form as
\begin{equation}
	\bx^* = \underset{\bx\in\lbrace 0,1\rbrace^n}{\arg\min} \ \bx^T \mathbf{P}(\alpha) \bx 
\end{equation}
\begin{equation}\label{eq:Qij}
	P_{ij}(\alpha) = R_{ij}-\alpha ( R_{ij} + \delta_{ij} I_i ),
\end{equation}
where $\delta_{ij} = 1_{\{i=j\}}$ denotes the Kronecker delta.

\subsection{Controlling the Number of Selected APs}
\label{control_aps}
In addressing our APs selection problem, the objective is to choose the minimum number of APs that yield a certain localization accuracy. This requirement can be incorporated by introducing a constraint to Eq~\ref{eq:qfs}, ensuring the selection of precisely $k$ APs. One common strategy is to append a Lagrangian penalty term~\cite{lagrange}, such as $\lambda\left(\left(\sum_ix_i\right)-k\right)^2$, to $Q(\bx,\alpha)$. This penalty term is only minimized when exactly $k$ APs are selected, where $\lambda$ controls how important it is to exactly match the constraint~\cite{lagrange}. Solving the problem yields the desired solution. However, two primary challenges emerge from this approach. First, determining an appropriate value for $\lambda$ is not a straightforward. Second, it is uncertain whether a manually-selected $k$ will lead to the best localization accuracy~\cite{FS}.

Facing these challenges, we present an alternative approach for determining the number of selected APs. Instead of relying on a penalty term, we exploit the intrinsic connection between the parameter $\alpha$, controls the balance between importance and redundancy, and the number of APs in the solution ($\lVert\bx^*\rVert_1$). This premise is substantiated by considering the extreme values of $\alpha \in [0,1]$. Specifically, setting $\alpha=0$ renders all diagonal entries of $\mathbf{P}(0)$ as $0$, emphasizing redundancy. Naturally, both the empty set of features and any single AP are minimally redundant, resulting in the optimal selection of APs being either $\lbrace\emptyset\rbrace$ or $\lbrace\lbrace i\rbrace; i\in \{1,2,3,\dots,n\}\rbrace$. Conversely, when $\alpha=1$, the problem becomes linear with solely negative coefficients, leading to the selection of all APs as the optimal solution. Consequently, by iteratively adjusting $\alpha$ from $0$ to $1$ in sufficiently small increments, we observe a monotonic increase of the number of selected APs in the optimal solution 
in steps of one from $0$ to $n$. This insight implies that we can modify $\alpha$ to obtain a subset of any desired size $k\in\lbrace 0,1,\dots,n\rbrace$ as we quantify in Section~\ref{sec:evaluation}. The following lemma proves that we can control the number of selected APs using $\alpha$.

\begin{proposition}\label{prop:alpha}
	For all $Q(\cdot, \alpha)$ defined as in Eq~\ref{eq:Q} and $k\in\lbrace 0,1,\dots,n\rbrace$, there exist an $\alpha\in [0,1]$ such that $\bx^*\in\arg\min_{\bx}Q(\bx,\alpha)$ and $\norm{\bx^*}_1=k$. Moreover, the number of selected APs $k$ monotonically increases with $\alpha$.
\end{proposition}

\begin{proof}
The proof can be found in Appendix~\ref{sec:proof}.
\end{proof}

This lemma shows that we do not need additional constraints on the QUBO instance to control the number of features present in the global optimum. Hence, we can determine a suitable value for $\alpha$ ($\alpha^*$) that results in the minimum number of selected APs ($k$) in $\cO(\log n)$ steps using binary search. 
It is important to note that the value of $\alpha^*$ may not be singular. The outlined algorithm is provided in Alg~\ref{alg:alpha}.

\begin{algorithm}
    \caption{Binary Search for $k$ Optimal APs}
    \SetAlgoLined
    \SetKwInOut{Input}{Input}
    \SetKwInOut{Output}{Output}
    \DontPrintSemicolon
    
    \Input{Redundancy Matrix $\bR$. \\
           Importance vector $\bI$. \\
           Number of APs $n$. \\
           Termination threshold $\epsilon$. \\
           Required Number of APs $k$.}
    
    \Output{Value of $\alpha^*$, Optimal value for $\alpha$. \\ 
            Selected APs $\bx^*$ with $\lVert\bx^*\rVert_1=k$.}

    // Initialization 

    $\bx \gets ones(n)$ 
    
    base\_acc $\gets$ Localizer($\bx$) \Comment{Localization with full dataset}
    
    acc $\gets $ base\_acc
    
    $a\gets 0.0$ \Comment{Initialize $\alpha$ range min value}
    
    $b\gets 1.0$\Comment{Initialize $\alpha$ range max value}

     // Search for Optimal APs
     
 \SetKwRepeat{Do}{do}{while}

     \Do{acc - prev\_acc $\geq \epsilon$}{ 
 
        prev\_acc $ \gets $ acc
        
        $\alpha\gets \frac{(a+b)}{2}$ \Comment{$\alpha$ is a floating point number.}
        
        $\bx^*\gets Q^*(\bx,\alpha)$ \Comment{Using quantum/classical  annealing}
        
	  $k'\gets \norm{\bx^*}_1$ \Comment{Number of selected APs}
   
        acc $\gets$ Localizer($\bx^*$) \Comment{Localization accuracy}
        
    \eIf{acc\ $ < $\ base\_acc}{
             $a\gets \alpha$ \Comment{Increase selected APs}
        }{
            
            $b\gets \alpha$
        }
        
     }

    $\alpha^* \gets \alpha$
    
    \Return{$\alpha^*$, $\bx^*$}
    \label{alg:alpha}
\end{algorithm}

\section{Evaluation}
\label{sec:evaluation}

In this section, we evaluate our quantum APs selection algorithm  and implement it on a real D-wave quantum hybrid solver which leverages the D-Wave Advantage quantum processing unit (QPU)~\cite{DwaveS, Dwave} with 5000 qubits~\cite{hs}. We start by describing our real testbed. Then, we show the effect of the balancing parameter $\alpha$. Finally, we compare the proposed quantum algorithm with the classical counterpart.

\subsection{Experiment Setup}
 The collected data cover a building with five floors of a public dataset~\cite{data_paper}. The floors area is around $110\textrm{m}^{2}$~\cite{data_paper}. The total number of APs is 520. We used the already deployed WiFi APs in the building as well as the overheard APs from the neighbors buildings. The total number of samples used in the fingerprint is 10000 samples. Without loss of generality, we used the Random Forest Classifier as a classification model for floor localization~\cite{random_forest}. The classical machine has a 2.2GHz CPU and 12GB RAM.
 
\subsection{Effect of Parameter $\alpha$}

Figure~\ref{alg:alpha} shows the effect of the balancing parameter $\alpha$ on the number of selected APs and the floor localization accuracy. The figure confirms that the  $\alpha$ parameter can control the number of selected APs as we proved in Section~\ref{control_aps}.

\begin{figure}[!t]
    \centering  \includegraphics[width=0.8\linewidth]{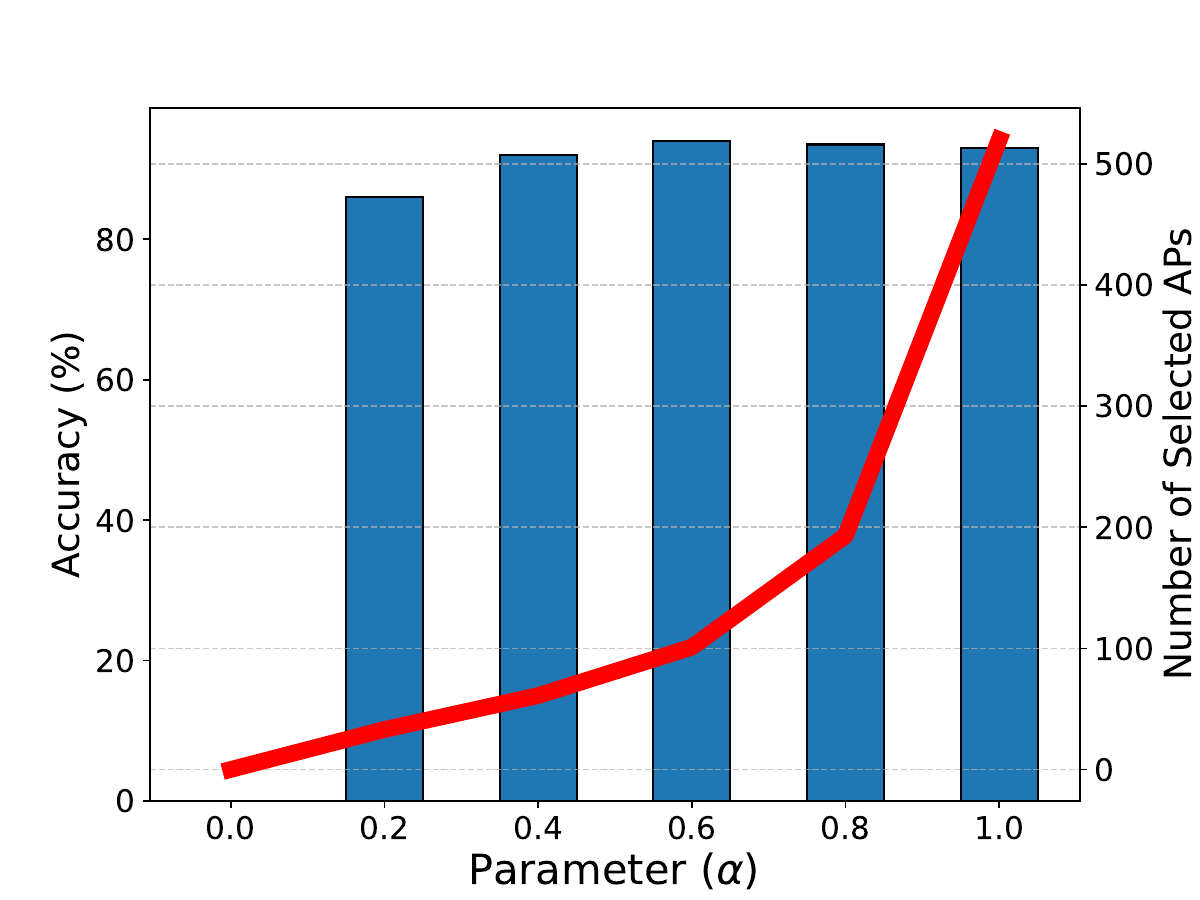}
    \caption{Effect of parameter $\alpha$ on the number of selected APs (red line) and floor localization accuracy (blue bars).} \label{fig:alpha}
\end{figure}

\begin{figure}[!t]
    \centering
\includegraphics[width=0.8\linewidth]{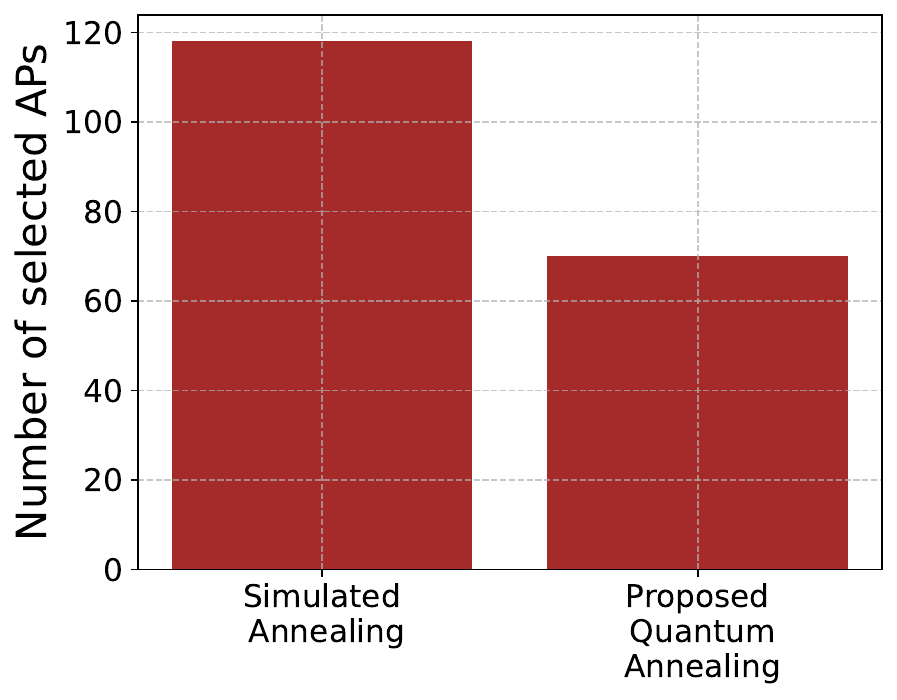}
    \caption{
    Selected number of APs.} \label{fig:s_aps}
\end{figure}

\begin{figure}[!t]
    \centering
\includegraphics[width=0.8\linewidth]{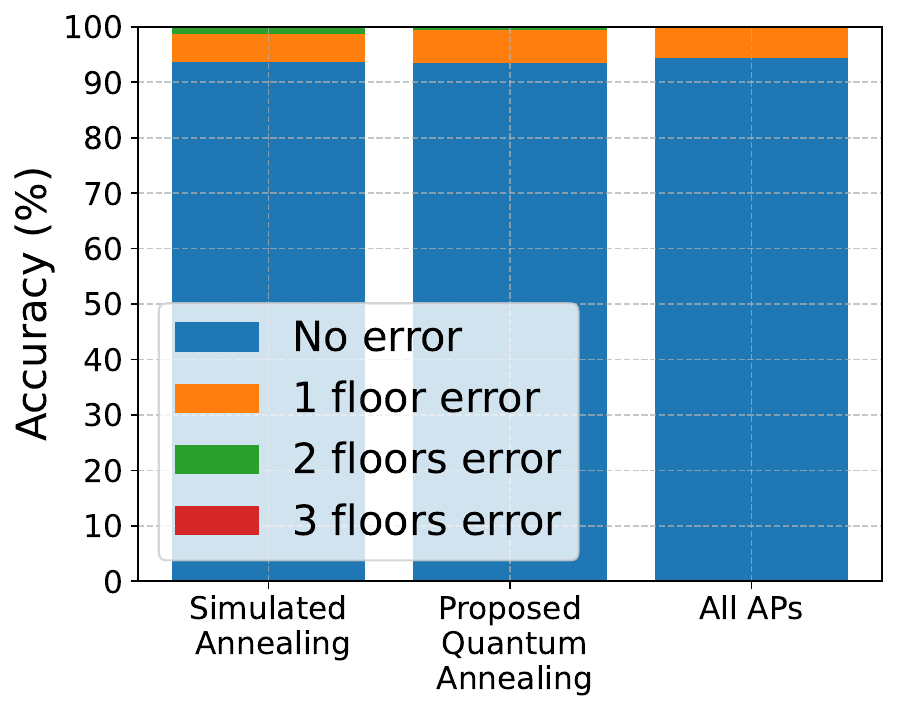}
    \caption{
    Floor localization accuracy using the classical simulated annealing algorithm and the quantum annealing algorithm.} \label{fig:acc}
\end{figure}

\begin{figure}[!t]
    \centering
\includegraphics[width=0.8\linewidth]{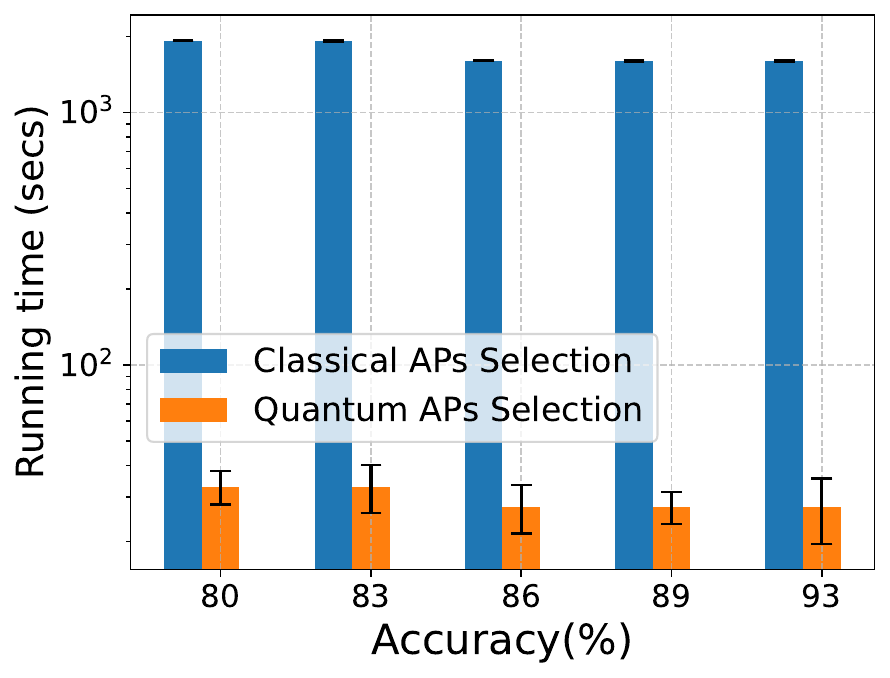}
    \caption{Running time in log-scale of classical and quantum APs selection  algorithms.} \label{fig:time}
\end{figure}

\subsection{Comparison with Classical APs
Selection}
In this section, we compare our proposed quantum annealing APs selection algorithm to the classical simulated annealing algorithm \cite{simulated} under different metrics: number of selected APs, accuracy, and running time. We end the section by showing the importance and redundancy of the selected APs for both algorithms.

\subsubsection{Selected number of APs}
Figure~\ref{fig:s_aps} shows the optimal number of APs obtained by running Algorithm~\ref{alg:alpha} using the classical simulated annealing algorithm and the proposed quantum annealing algorithm. The figure shows that we can achieve the same accuracy as the complete set of APs (520 APs) using 119 APs with simulated annealing  and only 70 APs  with the proposed quantum annealing algorithm. The figure also confirms that the proposed quantum annealing algorithm can reach better effective APs due to quantum tunneling \cite{tunneling}, where the quantum annealing process can penetrate through a potential barrier, enabling exploring the solution space efficiently.

\subsubsection{Accuracy}
Figure~\ref{fig:acc} illustrates the classical simulated annealing floor accuracy (using the selected 119 APs), the proposed quantum annealing accuracy (using the selected 70 APs), and the accuracy of using the complete set of APs (520 APs). Although the three sets of used APs achieve a comparable accuracy, the proposed quantum algorithm achieves a superior floor localization accuracy compared to the simulated annealing accuracy and even the accuracy while using all APs. This highlight the ability of the proposed quantum algorithm to select the set of effective APs that maximize the importance and minimize the redundancy. This helps also in removing noisy APs, which is one reason why the set of APs selected by the quantum annealing algorithm performs better than the full set of APs.   

\subsubsection{Running time}
Figure~\ref{fig:time} further shows the running time of  the classical and the quantum version of Algorithm~\ref{alg:alpha}. The figure highlights that the quantum APs selection algorithm can provide more than an order of magnitude saving in running time than the classical counterpart. 

\subsubsection{Importance and Redundancy}
Figure~\ref{fig:red_imp_full} shows the full set of APs importance and redundancy. The figure shows that the dataset has many redundant (color more towards yellow at \textit{off-diagonal} elements) and irrelevant APs (low importance). Figures~\ref{fig:imp_red_c} and~\ref{fig:imp_red_q} further show the quality of the selected APs using classical simulated annealing and quantum annealing, respectively. The figures confirms that the proposed quantum annealing algorithm capture more of the best APs. This is reflected by showing higher Importance of the selected APs and less yellowish color on the off-diagonal elements in the redundancy matrix.

\begin{figure}[!t]
  \begin{subfigure}{0.23\textwidth}
\includegraphics[width=\linewidth]{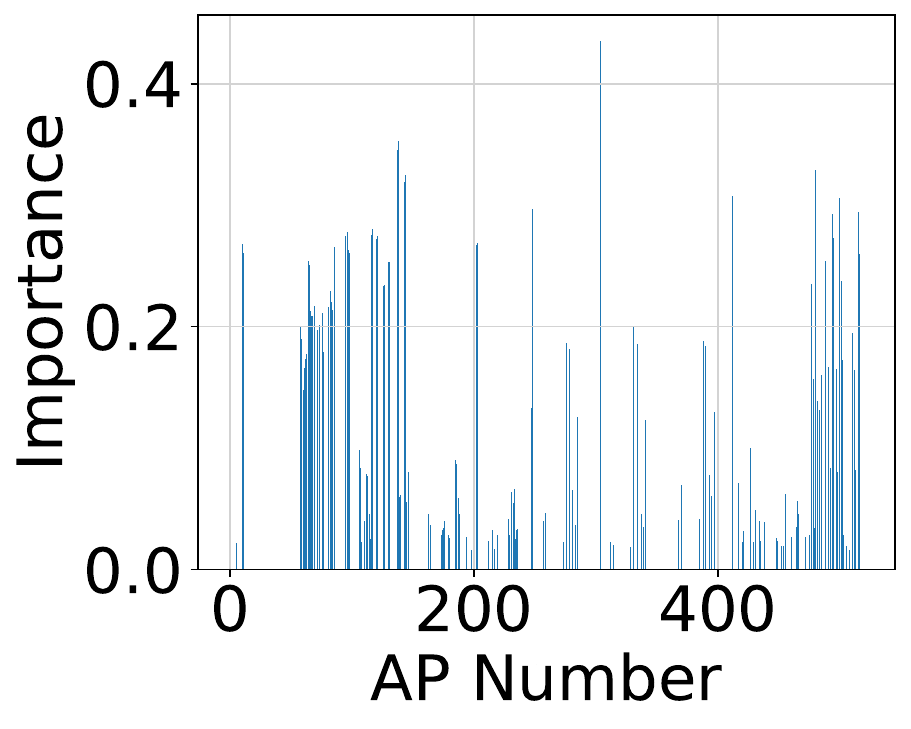}
    \caption{Importance} 
  \end{subfigure}%
  \hspace*{\fill}   
  \begin{subfigure}{0.23\textwidth}
    \includegraphics[width=\linewidth]{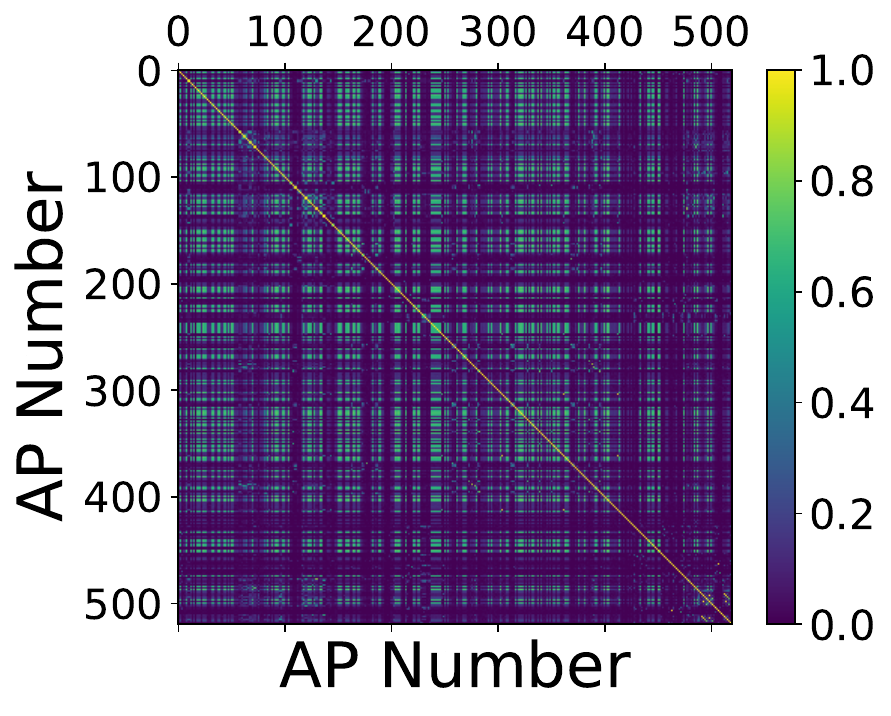}
    \caption{Redundancy} 
  \end{subfigure}%
  \hspace*{\fill}   %
\caption{Importance and Redundancy in all APs dataset (redrawn for convenience).} \label{fig:red_imp_full}
\end{figure}

\begin{figure}[!t]
  \begin{subfigure}{0.23\textwidth}
    \includegraphics[width=\linewidth]{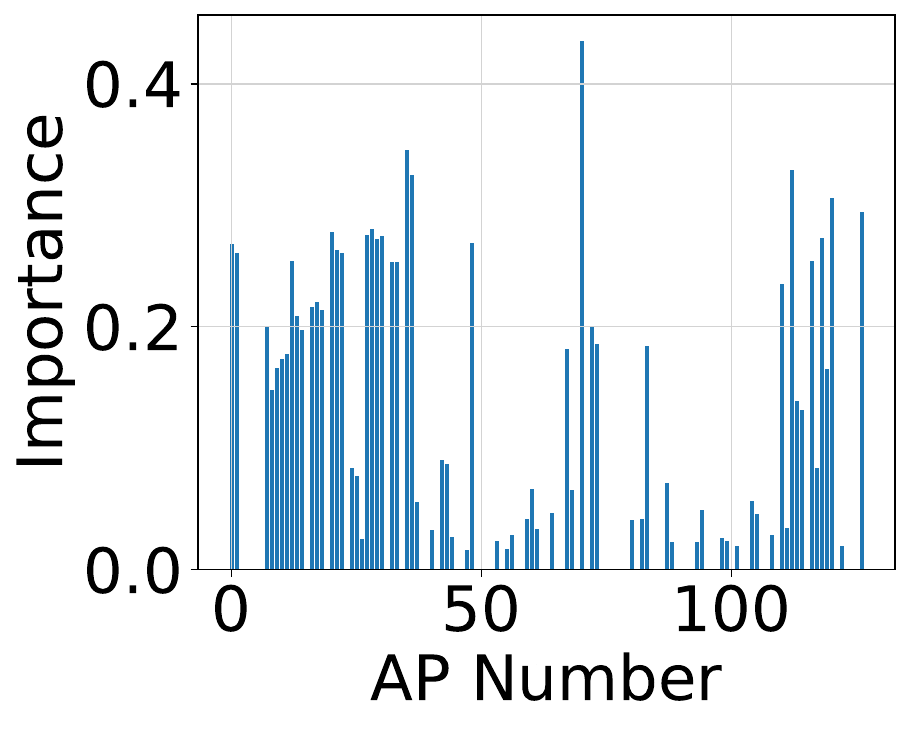}
    \caption{Importance} \label{fig:1a}
  \end{subfigure}%
  \hspace*{\fill}   
  \begin{subfigure}{0.23\textwidth}
    \includegraphics[width=\linewidth]{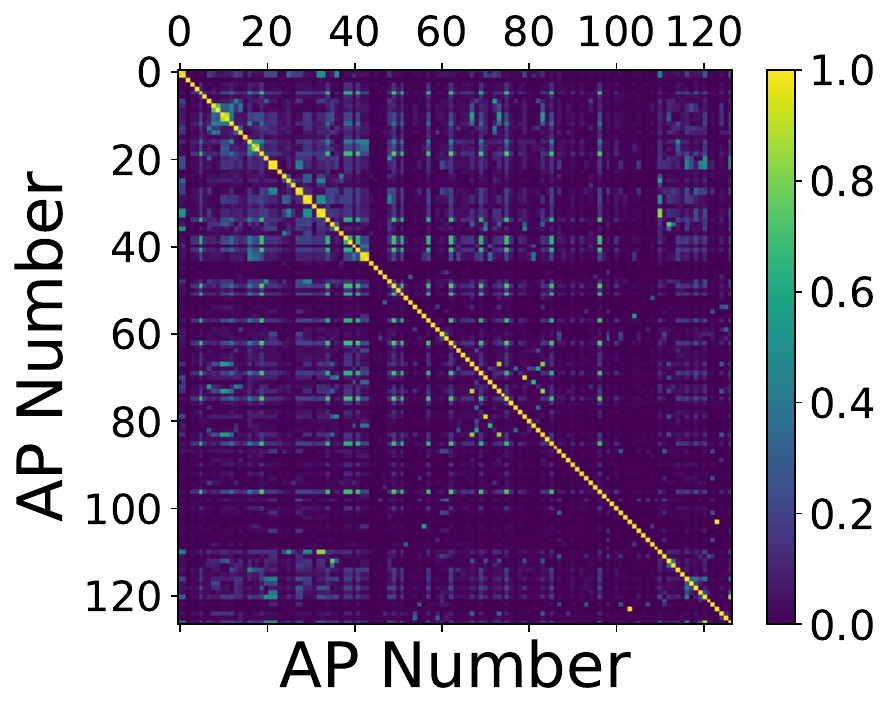}
    \caption{Redundancy} \label{fig:1b}
  \end{subfigure}%
  \hspace*{\fill}   %
\caption{Importance and Redundancy of the 120 selected APs by the simulated annealing.} \label{fig:imp_red_c}
\end{figure}

\begin{figure}[!t]
  \begin{subfigure}{0.23\textwidth}
    \includegraphics[width=\linewidth]{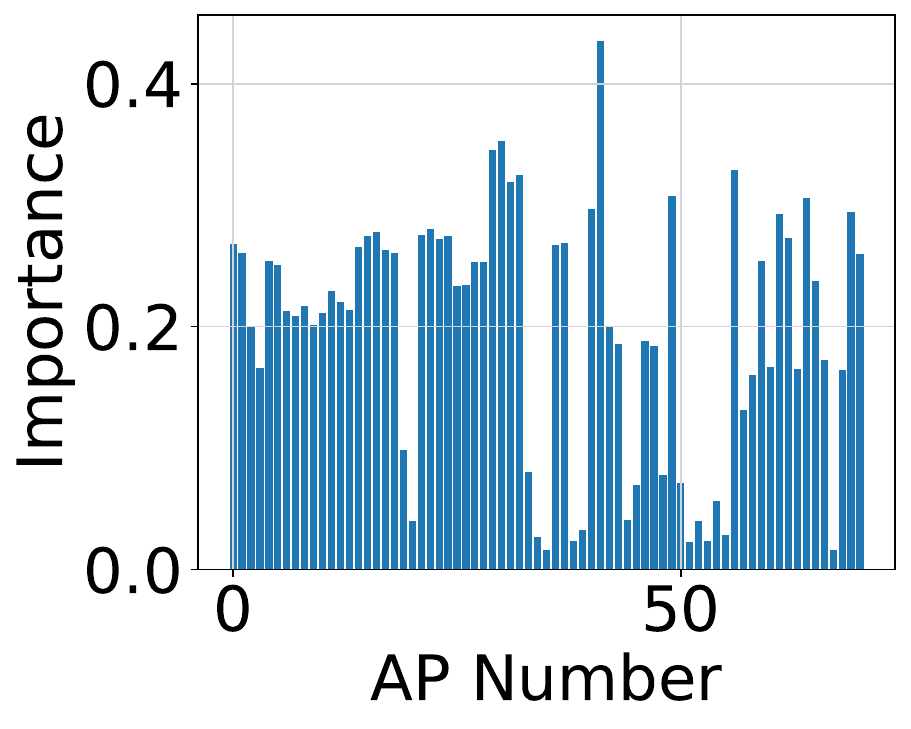}
    \caption{Importance} 
  \end{subfigure}%
  \hspace*{\fill}   
  \begin{subfigure}{0.23\textwidth}
    \includegraphics[width=\linewidth]{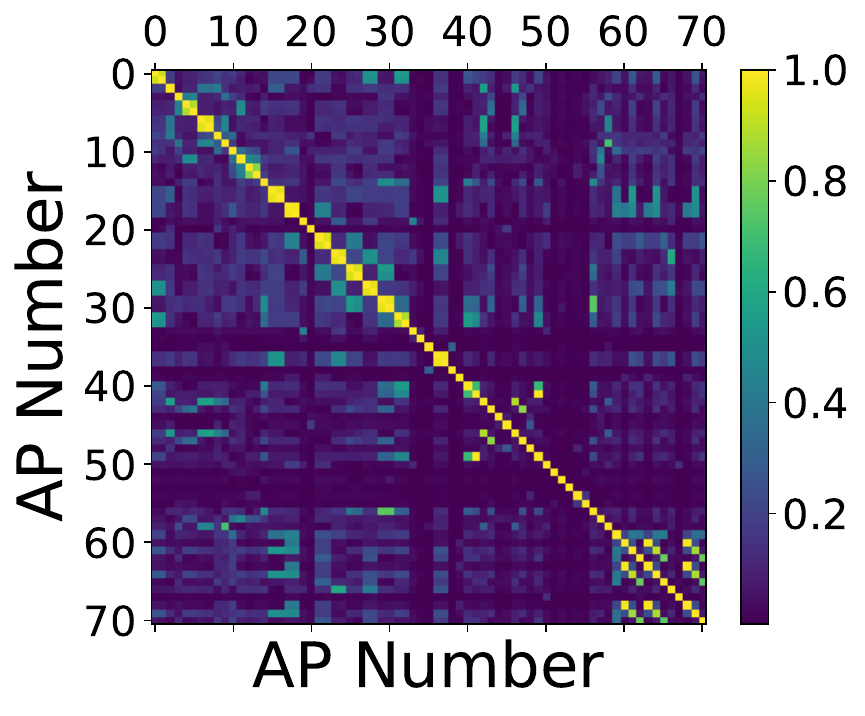}
    \caption{Redundancy} 
  \end{subfigure}%
  \hspace*{\fill}   %
\caption{Importance and Redundancy of the 70 selected APs by the proposed quantum algorithm. } \label{fig:imp_red_q}
\end{figure}

\section{Conclusion}
\label{sec:conclude}

In this paper, we introduced a quantum APs selection  algorithm for large-scale localization systems. The proposed quantum algorithm leverages quantum annealing to eliminate redundant and noisy APs from the fingerprint. We illustrated how to formulate the APs selection problem as a QUBO problem and how to encode the number of powerful APs constraint in the QUBO objective function. 
We deployed our quantum algorithm on a real D-Wave Systems quantum machine and evaluated its performance in a real testbed for floor localization. Our results showed that by selecting fewer than 14\% of the available APs, our quantum algorithm achieves the same floor localization accuracy compared to utilizing the entire set of APs and a better accuracy over utilizing the reduced dataset by the classical APs selection counterpart. Moreover, the proposed quantum algorithm achieved a speedup of more than an order of magnitude  over the corresponding classical APs selection algorithms which highlight the promise of the proposed quantum APs selection algorithm for large-scale localization.

\balance
\bibliographystyle{unsrt}
\bibliography{LCN24.bbl}

\begin{thebibliography}{10}

\bibitem{youssef2015towards}
Moustafa Youssef.
\newblock Towards truly ubiquitous indoor localization on a worldwide scale.
\newblock In {\em ACM SIGSPATIAL.}, pages 1--4, 2015.

\bibitem{shokry2017tale}
Ahmed Shokry, Moustafa Elhamshary, and Moustafa Youssef.
\newblock The tale of two localization technologies: Enabling accurate low-overhead wifi-based localization for low-end phones.
\newblock In {\em ACM SIGSPATIAL}, pages 1--10, 2017.

\bibitem{shokry2018deeploc}
Ahmed Shokry, Marwan Torki, and Moustafa Youssef.
\newblock Deeploc: a ubiquitous accurate and low-overhead outdoor cellular localization system.
\newblock In {\em ACM SIGSPATIAL.}, pages 339--348, 2018.

\bibitem{ibrahim2011cellsense}
Mohamed Ibrahim and Moustafa Youssef.
\newblock Cellsense: An accurate energy-efficient gsm positioning system.
\newblock {\em IEEE Transactions on Vehicular Technology}, 61(1):286--296, 2011.

\bibitem{jovic2015review}
Alan Jovi{\'c}, Karla Brki{\'c}, and Nikola Bogunovi{\'c}.
\newblock A review of feature selection methods with applications.
\newblock In {\em 2015 38th international convention on information and communication technology, electronics and microelectronics (MIPRO)}, pages 1200--1205. Ieee, 2015.

\bibitem{quantum_arx}
Ahmed Shokry and Moustafa Youssef.
\newblock {Challenge: Quantum Computing for Location Determination}.
\newblock {\em arXiv e-prints}, pages arXiv--2106, 2021.

\bibitem{quantum_vision}
Ahmed Shokry and Moustafa Youssef.
\newblock {Towards Quantum Computing for Location Tracking and Spatial Systems}.
\newblock In {\em Proceedings of the 29th International Conference on Advances in Geographic Information Systems}, pages 278--281, 2021.

\bibitem{device_indp_q}
Ahmed Shokry and Moustafa Youssef.
\newblock {Device-independent Quantum Fingerprinting for Large Scale Localization}.
\newblock In {\em 2022 MedComNet}, pages 208--215. IEEE, 2022.

\bibitem{SHOKRY2023}
Ahmed Shokry and Moustafa Youssef.
\newblock Quantum fingerprinting for heterogeneous devices localization.
\newblock {\em Computer Communications}, 2023.

\bibitem{quantum_lcn}
Ahmed Shokry and Moustafa Youssef.
\newblock {A Quantum Algorithm for RF-based Fingerprinting Localization Systems}.
\newblock {\em IEEE Conference on Local Computer Networks (LCN)}, 2022.

\bibitem{quantum_qce}
Ahmed Shokry and Moustafa Youssef.
\newblock {QLoc: A Realistic Quantum Fingerprint-based Algorithm for Large Scale Localization}.
\newblock {\em IEEE International Conference on Quantum Computing and Engineering (QCE)}, 2022.

\bibitem{shokry2023qradar}
Ahmed Shokry and Moustafa Youssef.
\newblock Qradar: A deployable quantum euclidean similarity large-scale localization system.
\newblock In {\em 2023 IEEE 48th Conference on Local Computer Networks (LCN)}, pages 1--8. IEEE, 2023.

\bibitem{zook2023quantum}
Yousef Zook, Ahmed Shokry, and Moustafa Youssef.
\newblock A quantum fingerprinting algorithm for next generation cellular positioning.
\newblock {\em arXiv preprint arXiv:2306.08108}, 2023.

\bibitem{schlosshauer2005decoherence}
Maximilian Schlosshauer.
\newblock Decoherence, the measurement problem, and interpretations of quantum mechanics.
\newblock {\em Reviews of Modern physics}, 76(4):1267, 2005.

\bibitem{QUBO1}
Gary Kochenberger, Jin-Kao Hao, Fred Glover, Mark Lewis, Zhipeng L{\"u}, Haibo Wang, and Yang Wang.
\newblock The unconstrained binary quadratic programming problem: a survey.
\newblock {\em Journal of combinatorial optimization}, 28:58--81, 2014.

\bibitem{QUBO2}
Fred Glover, Gary Kochenberger, and Yu~Du.
\newblock A tutorial on formulating and using qubo models.
\newblock {\em arXiv preprint arXiv:1811.11538}, 2018.

\bibitem{Dwave}
Sergio Boixo, Troels~F R{\o}nnow, Sergei~V Isakov, Zhihui Wang, David Wecker, Daniel~A Lidar, John~M Martinis, and Matthias Troyer.
\newblock Evidence for quantum annealing with more than one hundred qubits.
\newblock {\em Nature physics}, 10(3):218--224, 2014.

\bibitem{nielsen2002quantum}
Michael~A Nielsen and Isaac Chuang.
\newblock Quantum computation and quantum information, 2002.

\bibitem{qa}
Catherine~C McGeoch.
\newblock {\em Adiabatic quantum computation and quantum annealing: Theory and practice}.
\newblock Springer Nature, 2022.

\bibitem{finance}
Rom{\'a}n Or{\'u}s, Samuel Mugel, and Enrique Lizaso.
\newblock Quantum computing for finance: Overview and prospects.
\newblock {\em Reviews in Physics}, 4:100028, 2019.

\bibitem{nath2021review}
Rajdeep~Kumar Nath, Himanshu Thapliyal, and Travis~S Humble.
\newblock A review of machine learning classification using quantum annealing for real-world applications.
\newblock {\em SN Computer science}, 2:1--11, 2021.

\bibitem{griffiths2018introduction}
David~J Griffiths and Darrell~F Schroeter.
\newblock {\em Introduction to quantum mechanics}.
\newblock Cambridge university press, 2018.

\bibitem{ISING}
Fred Glover, Gary Kochenberger, and Yu~Du.
\newblock A tutorial on formulating and using qubo models.
\newblock {\em arXiv preprint arXiv:1811.11538}, 2018.

\bibitem{main_classical_paper}
Litao Han, Li~Jiang, Qiaoli Kong, Ji~Wang, Aiguo Zhang, and Shiming Song.
\newblock {Indoor Localization within Multi-Story Buildings Using MAC and RSSI Fingerprint Vectors}.
\newblock {\em Sensors}, 19(11):2433, May 2019.

\bibitem{fi_fingerprint_2018}
K.~Kim, S.~Lee, and K.~Huang.
\newblock {A scalable deep neural network architecture for multi-building and multi-floor indoor localization based on Wi-Fi fingerprinting.}
\newblock {\em Big Data Analytics}, 3(4), 2018.

\bibitem{fi_fingerprint_ka}
Khaled Alkiek, Aya Othman, Hamada Rizk, and Moustafa Youssef.
\newblock {Deep Learning-Based Floor Prediction Using Cell Network Information}.
\newblock In {\em Proceedings of the 28th International Conference on Advances in Geographic Information Systems}, SIGSPATIAL '20, page 663–664, New York, NY, USA, 2020. Association for Computing Machinery.

\bibitem{VIFI_fingerprint_2019}
Giuseppe Caso, Luca De~Nardis, Filip Lemic, Vlado Handziski, Adam Wolisz, and Maria-Gabriella~Di Benedetto.
\newblock {ViFi: Virtual Fingerprinting WiFi-Based Indoor Positioning via Multi-Wall Multi-Floor Propagation Model}.
\newblock {\em IEEE Transactions on Mobile Computing}, 19(6):1478--1491, 2020.

\bibitem{cramer1999mathematical}
Harald Cram{\'e}r.
\newblock {\em Mathematical methods of statistics}, volume~26.
\newblock Princeton university press, 1999.

\bibitem{pearson}
Philip Sedgwick.
\newblock Pearson’s correlation coefficient.
\newblock {\em Bmj}, 345, 2012.

\bibitem{lagrange}
Brian Beavis and Ian Dobbs.
\newblock {\em Optimisation and stability theory for economic analysis}.
\newblock Cambridge university press, 1990.

\bibitem{FS}
Sascha M{\"u}cke, Raoul Heese, Sabine M{\"u}ller, Moritz Wolter, and Nico Piatkowski.
\newblock Feature selection on quantum computers.
\newblock {\em Quantum Machine Intelligence}, 5(1):11, 2023.

\bibitem{DwaveS}
D-Wave Systems.
\newblock {\em Hybrid Solvers for Quadratic Optimization}.
\newblock 2022.

\bibitem{hs}
D-Wave Systems.
\newblock {\em D-Wave Hybrid Solver Service + Advantage: Technology Update}.
\newblock 2022.

\bibitem{data_paper}
Joaquín Torres-Sospedra, Raúl Montoliu, Adolfo Martínez-Usó, Joan~P. Avariento, Tomás~J. Arnau, Mauri Benedito-Bordonau, and Joaquín Huerta.
\newblock {UJIIndoorLoc: A new multi-building and multi-floor database for WLAN fingerprint-based indoor localization problems}.
\newblock In {\em 2014 International Conference on Indoor Positioning and Indoor Navigation (IPIN)}, pages 261--270, 2014.

\bibitem{random_forest}
Tin~Kam Ho.
\newblock Random decision forests.
\newblock In {\em Proceedings of 3rd international conference on document analysis and recognition}, volume~1, pages 278--282. IEEE, 1995.

\bibitem{simulated}
Peter~JM Van~Laarhoven, Emile~HL Aarts, Peter~JM van Laarhoven, and Emile~HL Aarts.
\newblock {\em Simulated annealing}.
\newblock Springer, 1987.

\bibitem{tunneling}
Siddharth Muthukrishnan, Tameem Albash, and Daniel~A Lidar.
\newblock Tunneling and speedup in quantum optimization for permutation-symmetric problems.
\newblock {\em Physical Review X}, 6(3):031010, 2016.

\end{thebibliography}

\begin{appendices}

\section{Proof of lemma ~\ref{prop:alpha}}
\label{sec:proof}

	Recall that the both the redundancy $R(\bx)$ and the importance $I(\bx)$ are $\geq 0$ due to non-negativity of the absolute value of the Pearson correlation and the Cramer's V measurements, respectively.
 
	When $\alpha=0$, the objective function in Eq~\ref{eq:Q}, $Q(\cdot, 0)$, tries to minimize the redundancy, leading to either the zero vector $\bar{\bm{0}}$ as the optimal solution or one of one-hot encoding vectors with the index of AP $i$ is one and zero otherwise, $\forall i\in \{1,2,3,\dots,n\}$ is optimal. This covers the cases of not selecting any AP ($k=0$) or selecting a single AP ($k=1$).
	
    When $\alpha=1$, the objective function in Eq~\ref{eq:Q}, becomes 
    \begin{equation}
        Q(\bx,1)=-I(\bx)=-\sum_{i\in \{1,2,\dots,n\}}I_ix_i
    \end{equation}
    which is minimized by a vector of ones  $\bar{\bm{1}}$, covering the case of selecting all APs ($k=n$).
	
    Now, for all $k\in\lbrace 0,\dots,n\rbrace$, consider the functions, \begin{equation}
		Q^*_{\leq k}(\alpha) := \min_{\bx\in\{0,1\}^n}Q_{\alpha}(\bx)~\text{s.t. } \norm{\bx}_1\leq k\;.
	\end{equation}
	These functions are piece-wise linear and strictly decreasing in $\alpha$~\cite{FS}, due to 
     \begin{equation}
         \frac{\partial Q_{\alpha}(\bx)}{\partial\alpha} = -(R(\bx)+I(\bx))\leq 0
     \end{equation}
      
	This implies further for all $\alpha\in [0,1]$ and $k\in \{1,2,3,\dots,n\}$ that \begin{align}
		\min_{\substack{\bx\in\{0,1\}^n \\ \norm{\bx}_1\leq k-1}} R(x) \leq \min_{\substack{\bx\in\{0,1\}^n \\ \norm{\bx}_1\leq k}} R(x) \\
		\max_{\substack{\bx\in\{0,1\}^n \\ \norm{\bx}_1\leq k-1}} I(x) \leq \max_{\substack{\bx\in\{0,1\}^n \\ \norm{\bx}_1\leq k}} I(x)\;.
	\end{align}
 
	Hence, for any $k<k'$, unless $Q^*_{\leq k}$ = $Q^*_{\leq k'}$,
\begin{equation}
    Q^*_{\leq k'}(\alpha)\leq Q^*_{\leq k}(\alpha)
\end{equation}
as a consequence of $Q^*_{\leq k}$ and $Q^*_{\leq k'}$ being non-increasing, from which follows the proof. If indeed $Q^*_{\leq k}$ = $Q^*_{\leq k'}$, both binary vectors $\bx$ and $\bx'$ with $\norm{\bx}_1=k$ and $\norm{\bx'}_1=k'$ would be optimal, from which the proof still follows. \qedsymbol{}

\end{appendices}

\end{document}